\documentclass[reqno]{amsart}

\usepackage{amsmath}
\usepackage{amssymb}
\usepackage{amsfonts}
\usepackage{mathtools}
\mathtoolsset{%
}

\usepackage[utf8]{inputenc}
\usepackage[T1]{fontenc}

\usepackage[dvipsnames,svgnames]{xcolor}
\colorlet{MyBlue}{DodgerBlue!60!Black}
\colorlet{MyGreen}{DarkGreen!85!Black}

\usepackage{subcaption}
\usepackage{tikz}
\usepackage{pgfplots}
\usetikzlibrary{calc,patterns}

\usepackage{acronym}
\usepackage{latexsym}
\usepackage{paralist}
\usepackage{wasysym}
\usepackage{xspace}

\usepackage[authoryear,compress]{natbib}

\usepackage{hyperref}
\hypersetup{
colorlinks=true,
linktocpage=true,
pdfstartview=FitH,
breaklinks=true,
pdfpagemode=UseNone,
pageanchor=true,
pdfpagemode=UseOutlines,
plainpages=false,
bookmarksnumbered,
bookmarksopen=false,
bookmarksopenlevel=1,
hypertexnames=true,
pdfhighlight=/O,
urlcolor=MyBlue!60!black,linkcolor=MyBlue!70!black,citecolor=DarkGreen!70!black, 
pdftitle={},
pdfauthor={},
pdfsubject={},
pdfkeywords={},
pdfcreator={pdfLaTeX},
pdfproducer={LaTeX with hyperref}
}

\def\EMAIL#1{\email{\href{mailto:#1}{\texttt #1}}}

\numberwithin{equation}{section}  
\usepackage[sort&compress,capitalize,nameinlink]{cleveref}
\crefname{example}{Ex.}{Exs.}

\crefrangeformat{equation}{\upshape(#3#1#4)\textendash(#5#2#6)}

\newcommand{\from}{\colon}

\newcommand{\R}{\mathbb{R}}

\DeclareMathOperator*{\argmin}{arg\,min}

\DeclareMathOperator{\ex}{\mathbb{E}}

\DeclareMathOperator{\prob}{\mathbb{P}}

\DeclarePairedDelimiter{\braces}{\{}{\}}

\DeclarePairedDelimiter{\abs}{\lvert}{\rvert}

\DeclarePairedDelimiterX{\braket}[2]{\langle}{\rangle}{#1\mathopen{}\hspace{1pt}\delimsize\vert\hspace{1pt}\mathopen{}#2}

\DeclarePairedDelimiterX{\inner}[2]{\langle}{\rangle}{#1,#2}
\DeclarePairedDelimiterX{\setdef}[2]{\{}{\}}{#1:#2}

\DeclarePairedDelimiterXPP{\probof}[1]{\prob}{(}{)}{}{%

#1}

\DeclarePairedDelimiterXPP{\exof}[1]{\ex}{[}{]}{}{%

#1}

\newcommand{\txs}{\textstyle}
\newcommand{\textpar}[1]{\textup(#1\textup)}

\newcommand{\BPR}{\text{\normalfont{BPR}}}

\usepackage[textwidth=30mm]{todonotes}
\usepackage{soul}
\setstcolor{red}
\sethlcolor{SkyBlue}

\usepackage[linesnumbered,ruled,vlined]{algorithm2e}

\theoremstyle{plain}
\newtheorem{theorem}{Theorem}
\newtheorem{corollary}[theorem]{Corollary}
\newtheorem*{corollary*}{Corollary}

\newtheorem{proposition}[theorem]{Proposition}

\theoremstyle{definition}
\newtheorem{definition}[theorem]{Definition}
\newtheorem*{definition*}{Definition}

\newtheorem*{hypothesis*}{Hypothesis}

\theoremstyle{remark}

\newtheorem*{remark*}{Remark}
\newtheorem*{notation*}{Notational remark}
\newtheorem{example}{Example}

\numberwithin{theorem}{section}
\numberwithin{remark}{section}
\numberwithin{example}{section}

\newcommand{\game}{\Gamma}

\newcommand{\play}{i}

\newcommand{\graph}{\mathcal{G}}
\newcommand{\vertices}{\mathcal{V}}
\newcommand{\edges}{\mathcal{E}}

\newcommand{\edge}{e}

\newcommand{\vertex}{v}
\newcommand{\vertexalt}{\vertex'}

\newcommand{\source}{o}

\newcommand{\sink}{d}

\newcommand{\pair}{\play}

\newcommand{\pairs}{\mathcal{I}}
\newcommand{\rate}{\mu}

\newcommand{\flow}{f}
\newcommand{\flows}{\mathcal{F}}

\newcommand{\load}{x}

\newcommand{\nRoutes}{P}
\newcommand{\routes}{\mathcal{\nRoutes}}
\newcommand{\route}{p}
\newcommand{\routealt}{\route'}

\newcommand{\toll}{\tau}

\newcommand{\interc}{t}
\newcommand{\slope}{a}

\newcommand{\const}{\lambda}
\newcommand{\increm}{\gamma}
\newcommand{\diffOD}{\Delta}
\newcommand{\degree}{\beta}

\newcommand{\cost}{c}

\newcommand{\eq}[1]{#1^{\ast}}
\newcommand{\opt}[1]{\tilde{#1}}

\DeclareMathOperator{\Eq}{Eq}
\DeclareMathOperator{\Opt}{Opt}
\DeclareMathOperator{\PoA}{PoA}

\newcommand{\obj}{L}

\title{Demand-Independent Optimal Tolls}

\author
[R.~Colini-Baldeschi]
{Riccardo Colini-Baldeschi}
\address{Core Data Science Group, Facebook Inc.,
1 Rathbone Place, London, W1T~1FB, UK,}
\EMAIL{rickuz@fb.com}

\author
[M.~Klimm]
{Max Klimm}
\address{JP f\"ur Operations Research,
Wirtschaftswissenschaftliche Fakult\"at,
Humboldt-Universit\"at zu Berlin,
Spandauer Straße 1,
10099 Berlin, Germany}
\EMAIL{max.klimm@hu-berlin.de}

\author
[M.~Scarsini]
{Marco Scarsini}
\address{Dipartimento di Economia e Finanza, LUISS, Viale Romania 32, 00197 Roma, Italy,}
\EMAIL{marco.scarsini@luiss.it}

\thanks{Riccardo Colini-Baldeschi and Marco Scarsini are members of GNAMPA-INdAM.
Max Klimm gratefully acknowledges the support and hospitality of LUISS during a visit in which this research was initiated.}


\newcommand{\acdef}[1]{\textit{\acl{#1}} \textpar{\acs{#1}}\acused{#1}}

\newacro{NE}{Nash equilibrium}
\newacroplural{NE}[NE]{Nash equilibria}
\newacro{WE}{Wardrop equilibrium}
\newacroplural{WE}[WE]{Wardrop equilibria}
\newacro{SO}{social optimum}
\newacro{DAMG}{directed acyclic multi-graph}

\newacro{KKT}{Karush\textendash Kuhn\textendash Tucker}
\newacro{OD}[O/D]{origin-destination}
\newacro{PoA}{price of anarchy}
\newacro{BPR}{bureau of public roads}

\newacro{DIOT}{demand-independent optimum toll}

\begin{document}


\maketitle

\begin{abstract}

Wardrop equilibria in nonatomic congestion games are in general inefficient as they do not induce an optimal flow that minimizes the total travel time. Network tolls are a prominent and popular way to induce an optimum flow in equilibrium. The classical approach to find such tolls is marginal cost pricing which requires the exact knowledge of the demand on the network. In this paper, we investigate under which conditions \aclp{DIOT} exist that induce the system optimum flow for any travel demand in the network. We give several characterizations for the existence of such tolls both in terms of the cost structure and the network structure of the game. Specifically we show that \aclp{DIOT} exist if and only if the edge cost functions are shifted monomials as used by the Bureau of Public Roads. Moreover, non-negative \aclp{DIOT} exist when the network is a directed acyclic multi-graph. Finally, we show that any network with a single \acl{OD} pair admits \aclp{DIOT} that, although not necessarily non-negative, satisfy a budget constraint.
\end{abstract}

\section{Introduction}
\label{sec:introduction}


The impact of selfish behavior on the efficiency of traffic networks is a longstanding topic in the algorithmic game theory and operations research literature. Already more than half a century ago, \citet{War:PICE1952} stipulated a main principle of a traffic equilibrium that---in light of the omnipresence of modern route guidance systems---is as relevant as ever: ``The journey times on all the routes actually used are equal, and less than those which would be experienced by a single vehicle on any unused route.'' This principle can be formalized by modeling traffic as a flow in a directed network where edges correspond to road segments and vertices correspond to crossings. Each edge is endowed with a cost function that maps the total amount of traffic on it to a congestion cost that each flow particle traversing the edge has to pay. 
Further, we are given a set of commodities, each specified by an origin, a destination, and a flow demand. In this setting, a Wardrop equilibrium is a multi-commodity flow such that for each commodity the total cost of any used path is not larger than the total cost of any other path linking the commodity's origin and destination. Popular cost functions, put forward for the use in traffic models by the \cite{BPR} are of the form
\begin{align}
\label{eq:bpr}
\cost_\edge(\load_\edge) = \interc_\edge \biggl(1 +\alpha \Bigl( \frac{\load_\edge}{k_\edge} \Bigr)^{\!\degree} \biggr),
\end{align}
where $x_e$ is the traffic flow along edge $e$, $t_e \geq 0$ is the free-flow travel time, $k_e$ is the capacity of edge $e$ and $\alpha$ and $\beta$ are parameters fitted to the model.

While Wardrop equilibria are guaranteed to exist under mild assumptions on the cost functions \citep{BecMcGWin:Yale1956}, it is well known that they are inefficient in the sense that they do not minimize the overall travel time of all commodities \citep{Pig:Macmillan1920}. A popular mechanism to decrease the inefficiency of selfish routing are congestion tolls. A toll is a payment that the system designer defines for each edge of the graph and that has to be paid by each flow particle traversing the edge. By carefully choosing the edges' tolls, the system designer can steer the Wardrop equilibrium in a favorable direction. A classic approach first due to \citet{Pig:Macmillan1920} is \emph{marginal cost pricing} where the toll of each edge is equal to difference between the marginal social cost and the marginal private cost of the system optimum flow on that edge. Marginal cost pricing induces the system optimal flow---the one that minimizes the overall travel time---as a Wardrop equilibrium \citep{BecMcGWin:Yale1956}. Congestion pricing is not only an interesting theoretical issue that links system optimal flows and traffic equilibria, but also a highly relevant tool in practice, as various cities of the world, including Stockholm, Singapore, Bergen, and London, implement congestion pricing schemes to mitigate congestion \citep{gomez1994road,small1997road}.

\subparagraph{The problem}
\label{sec:problem}

Marginal cost pricing is an elegant way to induce the system optimum flow as a Wardrop equilibrium, but the concept crucially relies on the exact knowledge of the travel demand. As an example consider the Pigou network in \cref{fig:pigou} for an arbitrary flow demand of $\rate > 0$ going from $\source$ to $\sink$. The optimal flow only uses the lower edge with cost function $c(x) = x$ as long as $\rate \leq 1/2$. For demands $\rate > 1/2$, a flow of $1/2$ is sent along the lower edge and the remaining flow of $\rate - 1/2$ is sent along the upper edge. Using marginal cost pricing the toll of the lower edge is $\min\{\rate,1/2\}$ and no toll is to be payed for the upper edge, see \cref{fig:pigou_pricing}.

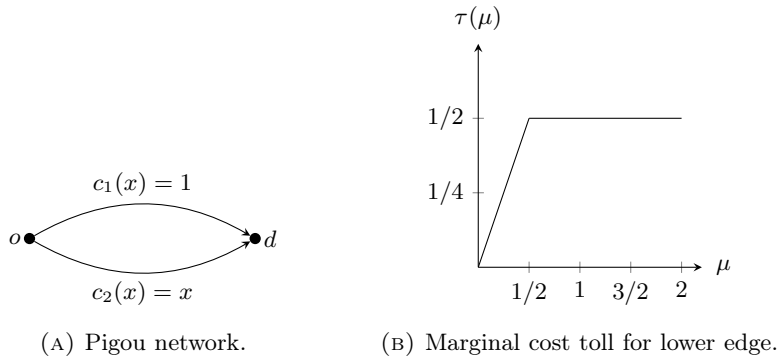
\begin{figure}
\small
\begin{subfigure}[b]{0.45\textwidth}
\centering
\begin{tikzpicture}[edgestyle/.style={-},>=stealth]
\tikzstyle{every circle node}=[fill,inner sep=1.5pt,outer sep=0pt]
\draw (0,0) node(s)[circle]{} node[left]{$\source$};
\draw (3,0) node(t)[circle]{} node[right]{$\sink$};
\draw[->] (s) to[bend left]  node[pos=0.5,above] {$c_1(x)=1$} (t);
\draw[->] (s) to[bend right] node[pos=0.5,below] {$c_2(x)=x$} (t);
\end{tikzpicture}
\subcaption{Pigou network.\label{fig:pigou}}
\end{subfigure}
\begin{subfigure}[b]{0.45\textwidth}
\centering
\begin{tikzpicture}[edgestyle/.style={-},>=stealth]
\begin{axis}[
width=0.8\textwidth,
height=0.8\textwidth,
axis lines=center,
xlabel={$\rate$},
every axis x label/.style={
    at={(ticklabel* cs:1.025)},
    anchor=west,
},
every axis y label/.style={
    at={(ticklabel* cs:1.025)},
    anchor=south,
},
ylabel={$\tau(\rate)$},
xmin=0, xmax=2.2,
ymin=0, ymax=0.75,
xtick={0,0.5,1,1.5,2},
xticklabels={$0$,$1/2$,$1$,$3/2$,$2$},
ytick={0,0.25,0.5},
yticklabels={$0$,$1/4$,$1/2$},
]
\addplot[domain=0:0.5]{x};
\addplot[domain=0.5:2]{0.5};
\end{axis}
\end{tikzpicture}
\subcaption{Marginal cost toll for lower edge.\label{fig:pigou_pricing}}
\end{subfigure}
\caption{Dependence of the marginal cost tolls for the Pigou network.}
\end{figure}

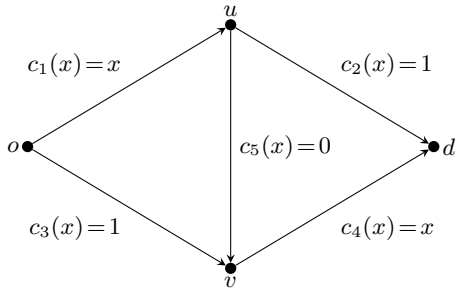
\begin{figure}
\small
\centering
\begin{tikzpicture}[scale=1.35,edgestyle/.style={-},>=stealth]
\tikzstyle{every circle node}=[fill,inner sep=1.5pt,outer sep=0pt]
\draw (-2,0) node(s)[circle]{} node[left]{$\source$};
\draw (0,1.2) node (v1)[circle]{} node[above]{$u$};
\draw (0,-1.2) node (v2)[circle]{} node[below]{$v$};
\draw (2,0) node(t)[circle]{} node[right]{$\sink$};
\draw[->] (s)  to node[pos=0.5,above left] {$c_1(x) \!=\! x$} (v1);
\draw[->] (v1) to node[pos=0.5,above right] {$c_2(x) \!=\! 1$} (t);
\draw[->] (s)  to node[pos=0.5,below left] {$c_3(x) \!=\! 1$} (v2);
\draw[->] (v2) to node[pos=0.5,below right] {$c_4(x) \!=\! x$} (t);
\draw[->] (v1) to node[pos=0.5,right] {$c_5(x) \!=\! 0$} (v2);
\end{tikzpicture}

\caption{Braess network.\label{fig:braess}}
\end{figure}

This toll scheme is problematic as it depends on the flow demand $\rate$ in the network. In particular, when the demand is estimated incorrectly, the resulting tolls will be sub-optimal. Assume the network designer expects a flow of $\rate=1/4$ and, thus, sets a (marginal cost) toll of $\tau_2 = 1/4$ on the lower link. When the actual total flow demand is higher than expected and equal to $\rate=1$, a fraction $1/4$ of the flow uses the upper edge and $3/4$ of the flow uses the lower edge resulting in a cost of $1/4+9/16$. However, it is optimal to induce an equal split between the edge which can be achieved by a toll of $1/2$ on the lower edge. Quite strikingly, the toll of $\tau_2 = 1/2$ is optimal for all possible demand values as it always induces the optimal flow: For demands less than $1/2$, a toll of $1/2$ on the lower edge does not hinder any flow particle from using the lower edge, which is optimal. On the other hand, for any demand larger $1/2$, the toll of $1/2$ on the lower edge forces the flow on this edge to not exceed $1/2$, which is optimal as well.

The situation is even more severe for the Braess network in \cref{fig:braess}. When the system designer expects a traffic demand of $1$ going from $\source$ to $\sink$, marginal cost pricing fixes a toll of $1$ on both the upper left and the lower right edge (both with cost function $c(x) = x$). When the demand is lower than expected, say, $\rate=1/2$, under marginal cost pricing, the flow is split equally between the lower and the upper path leading to a total cost of $5/4$. The optimal flow with flow value $1/2$, however, only uses the zig-zag-path $\source \to u \to v \to \sink$ with cost $1$. It is interesting to note that this flow is actually equal to the Wardrop equilibrium without any tolls. To conclude this example, marginal cost pricing may actually \emph{increase} the total cost of the Wardrop equilibrium when the travel demand is estimated incorrectly. We note that also in the Braess graph, there is a distinct toll vector that enforces the optimum flow as a Wardop equilibrium \emph{for any demand}. We will see that by setting a toll of $1$ on the central edge from $u$ to $v$, the Wardrop equilibrium for any flow demand is equal to the respective optimum flow.

We conclude that for both the Pigou network and the Braess network, marginal cost pricing is not robust with respect to changes in the demand since wrong estimates of the travel demand lead to sub-optimal tolls. Since such changes may occur frequently in road networks (e.g., due to sudden weather changes, accidents, or other unforeseen events), marginal cost pricing does not use the full potential of congestion pricing and may even be harmful for the traffic.
On the other hand, for both networks, there exist tolls that enforce the optimum flow as an equilibrium for \emph{any} flow demand.  In this paper, we systematically study conditions for such demand-independent optimal tolls to exist.
\subparagraph{Our results}
\label{sec:results}

In this paper we study the existence of \acp{DIOT} that induce the optimum flow as an equilibrium for \emph{any} flow demand. We give a precise characterization in terms of the cost structure on the edges for \acp{DIOT} to exists. Specifically, we show that \acp{DIOT} exist for any network where the cost of each edge is a BPR-type function, i.e.,
\begin{align}
\label{eq:bpr2}
\cost_\edge(\load) = \interc_\edge + \slope_\edge\load_\edge^\degree \quad \text{ for all $\edge \in \edges$}, 
\end{align}
where $\interc_\edge, \slope_\edge \in \R_+$ are arbitrary while $\degree \in \R_+$ is a common constant for all edges $\edge \in \edges$. This existence result holds regardless of the topology of the network and on the number of \ac{OD} pairs. On the other hand, for any cost function that is not of the form as in \eqref{eq:bpr2}, there is a simple network consisting of two parallel edges with cost function $c$ and cost function $c+b$ for some $b > 0$ that does not admit a \ac{DIOT}. Our existence result for networks with BPR-type cost functions is proven in terms of a characterization that uniquely determines the sum of the tolls along each path that is used by the optimum flow for some demand.

In general the \acp{DIOT} used in the characterization may use negative tolls as well. We provide an example of a network with BPR-type cost functions where a non-negative \ac{DIOT} does not exist. Besides conditions on the costs, conditions on the network are needed to guarantee the existence of non-negative \acp{DIOT}.  We show that non-negative \acp{DIOT} exist for \acp{DAMG} with BPR-type cost functions, like the Pigou network and the Braess network discussed in \S~\ref{sec:introduction}. This condition on the network is sufficient, but not necessary for the existence of non-negative \acp{DIOT}. 

Under a weaker condition than \ac{DAMG}, we prove the existence of \acp{DIOT} that follow a budget constraint of non-negativity of the total amount of tolls. This condition is satisfied by networks with a single \ac{OD} pair.

Due to space constraint, some of the proofs are deferred to the appendix.


\subparagraph{Related work}

Marginal cost pricing as a means to reduce the inefficiency of selfish resource allocation was first proposed by \citet{Pig:Macmillan1920} and subsequently discussed by \citet{Kni:QJE1924}. \citet{War:PICE1952} introduced the notion of a traffic equilibrium where each flow particle only uses shortest paths. \citet{BecMcGWin:Yale1956} showed that marginal cost pricing always induces the system optimal flow as a Wardrop equilibrium. The set of feasible tolls that induce optimal flows was explored by \citet{BerHeaRam:1997,HeaRam:Klu1998,LarPat:TR1999}. They showed that the set of optimal tolls can be described by a set of linear equations and inequalities.

This characterization led to various developments regarding the optimization of secondary objectives of the edge tolls, such as the minimization of the tolls collected from the users \citep{BaiHeaLaw:Networks2004,Dia:TR1999,Dia:TR2000}, or the minimzation of the number of edges that have positive tolls \citep{BaiRub:OR2009,BaiHeaLaw:JGO2010}. A problem closely related to the latter is to compute tolls for a given subset of edges with the objective to minimize the total travel time of the resulting equilibrium. \citet{HoeOlbSko:WINE2008} showed that this problem is NP-hard for general networks, and  gave an efficient algorithm for parallel edges graphs with affine cost functions. \citet{HarKleKliMoe:Network2015} generalized their result to arbitrary cost functions satisfying a technical condition. \citet{BonSalSch:SAGT2011} studied generalizations of this problem with further restrictions on the set of feasible edge tolls.

For heterogenous flow particles that trade off money and time differently, marginal cost pricing cannot be applied to find tolls that induce the system optimum flow. In this setting, \citet{ColDodRou:STOC2003} showed the existence of a set of tolls enforcing the system optimal flow, when there is a single commodity in the network. Similarly, \citet{YanHua:TR2004} studied how it is possible to design toll structure when there are users with different toll sensitivity. \citet{Fle:TCS2005} showed that in single source series-parallel networks the tolls have to be linear in the latency of the maximum latency path. \citet{KarKol:LNCS2006} and \citet{FleJaiMah:FOCS2004} independently generalized this result to arbitrary networks.
\citet{HanLoSunYan:EJOR2008} extended the previous results to different classes of cost functions.

Most of the literature assumes that the charged tolls cause no disutilty to the network users. For the case where tolls contribute to the cost, \citet{ColDodRou:JCSS2006} showed that marginal cost tolls do not improve the equilibrium flow for a large class of instances, including all instances with affine costs. They further showed that for these networks it is NP-hard to approximate the minimal total cost that can be achieved as a Wardrop equilibrium with tolls. 
\citet{KarKol:LNCS2005} proved that the total disutility due to taxation is bounded with respect to the social optimum for large classes of latency functions. 
Moreover, they showed that, if both the tolls and the latency are part of the social cost, then for some latency functions the coordination ratio improves when taxation is used.
For networks of parallel edges, \citet{ChrMehPyr:Algorithmica2014} studied a generalization of edge tolls where cost functions are allowed to increase in an arbitrary way. They showed that for affine cost functions, the \acl{PoA} is strictly better than in the original network, even when the demand is not known.

\citet{BroMar:P55IEEECDC2016,BroMar:IEEETAC2017} studied how marginal tolls can create perverse incentives when users have different sensitivity to the tolls and how it possible to circumvent this problem.
\citet{CarKakKan:Springer2006} studied the optimal toll problem for atomic congestion games. They proved that in the atomic case the optimal system performance cannot be achieved even in very simple networks. On the positive side they shown that there is a way to assign tolls to edges such that the induced social cost is within a factor of $2$ to the optimal social cost.
\citet{singh2008marginal} observed that marginal tolls weakly enforce optimal flows.
\citet{FotSpi:IM2008} showed that in series-parallel networks with increasing cost functions the optimal social cost can be induced with tolls.
\citet{MeiPar:9IWATT2016} discussed how in atomic congestion games with marginal tolls multiple equilibria are near-optimal when there is a large number of players.

In \citet{San:RES2002,San:JET2007} the problem is studied from a mechanism design perspective where the social planner has no information over the preferences of the users and has limited ability to observer the users' behavior.

\section{Model and preliminaries}
\label{sec:model}

In this section, we present some notation and basic definitions that are used in the sequel. We start with the underlying network model and will then introduce equilibria and tolls.

\subparagraph{Network model}

We consider a finite directed multi-graph $\graph=(\vertices,\edges)$ with vertex set $\vertices$ and edge set $\edges$. We call $(\vertex\to\vertexalt)$ the set of all edges $\edge$ whose tail is $\vertex$ and whose head is $\vertexalt$.
We  assume that there is a finite set 
of  \acdef{OD} pairs $\pair\in\pairs$, each with an individual \emph{traffic demand} $\rate^{\pair}\geq 0$  that has to be routed from 
an origin $\source^{\pair}\in\vertices$ to a destination $\sink^{\pair}\in\vertices $ via $\graph$. Denote the demand vector by $\boldsymbol{\rate}=(\rate^{\pair})_{\pair\in\pairs}$.
We call $\routes^{\pair}$  the set of (simple) paths joining $\source^{\pair}$ to $\sink^{\pair}$, where each path $\route\in\routes^{\pair}$ is a finite sequence of edges such that the head of each edge meets the tail of the subsequent edge.
For as long as all pairs $(\source^{\pair},\sink^{\pair})$ are different, the sets $\routes^{\pair}$ are  disjoint.
Call $\routes := \bigcup_{\pair\in\pairs} \routes^{\pair}$ the union of all such paths. 

Each path $\route$ is traversed by a \emph{flow} $\flow_{\route}\in\R_{+}$. Call $\boldsymbol{\flow} = (\flow_{\route})_{\route\in\routes}$ the vector of flows in the network.
The set of feasible flows for $\boldsymbol{\rate}$ is defined as
\begin{equation}
\label{eq:flows}
\txs
\flows(\boldsymbol{\rate})
	= \setdef*{\boldsymbol{\flow}\in\R_{+}^{\routes}}{\text{$\sum_{\route\in\routes^{\pair}} \flow_{\route} = \rate^{\pair}$ for all $\pair\in\pairs$}}.
\end{equation}
In turn, a routing flow $\boldsymbol{\flow}\in\flows(\boldsymbol{\rate})$ induces a \emph{load} on each edge $\edge\in\edges$ as
\begin{equation}
\label{eq:load}
\load_{\edge}
	= \sum_{\route\ni\edge} \flow_{\route}.
\end{equation}
We call $\boldsymbol{\load} = (\load_{\edge})_{\edge\in\edges}$ the corresponding \emph{load profile} on the network.
For each $\edge\in\edges$ consider a nondecreasing, continuous \emph{cost function} $\cost_{\edge}\from \R_+ \to \R_+$. Denote $\boldsymbol{\cost}=(\cost_{\edge})_{\edge\in\edges}$.
If $\boldsymbol{\load}$ is the load profile induced by a feasible routing flow $\boldsymbol{\flow}$, then the incurred delay on edge $\edge\in\edges$ is given by $\cost_{\edge}(\load_{\edge})$;
hence, with a slight abuse of notation, the associated cost of path $\route\in\routes$ is given by the expression $\cost_{\route}(\boldsymbol{\flow})
	\equiv \sum_{\edge\in\route} \cost_{\edge}(\load_{\edge})$.
We call the tuple $\game = (\graph,\pairs,\boldsymbol{\cost})$  a \emph{\textpar{nonatomic} routing game}.

\subparagraph{Equilibrium Flows and Optimal Flows}

A routing flow $\eq{\boldsymbol{\flow}}$ is a \acdef{WE} of $\game$ if, for all $\pair\in\pairs$, we have:
\begin{align*}
\cost_{\route}(\eq{\boldsymbol{\flow}})
	\leq \cost_{\routealt}(\boldsymbol{\eq{\flow}})
	\quad
	\text{for all $\route,\routealt\in\routes^{\pair}$ such that $\eq{\flow}_{\route}>0$}.
\end{align*}
This concept was introduced by \cite{War:PICE1952}. \cite{BecMcGWin:Yale1956} showed that Wardrop equilibria are the optimal solutions to the convex optimization problem
\begin{align*}
\min\; \sum_{\edge \in \edges} &\int_{0}^{x_e}	\cost_e(s) \,\text{d}s \\
\text{s.t.: }\; \load_e &= \sum_{\route \ni \edge} \flow_\route\\
\boldsymbol{\flow} &\in \flows,
\end{align*}
and, thus, are guaranteed to exist. A \acdef{SO} is a flow that minimizes the overall travel time, i.e., it solves the following total cost minimization problem:
\begin{align}
\begin{split}
\label{eq:SO}
\min\;	\obj(\boldsymbol{\flow}) &= \sum_{\route\in\routes} \flow_{\route} \cost_{\route}(\boldsymbol{\flow}),\\
\text{s.t.:}\quad \boldsymbol{\flow} &\in\flows.
\end{split}
\end{align}
A shown by \cite{BecMcGWin:Yale1956}, all Wardrop equilibria have the social cost. We write 
$\Eq(\game)
	= \obj(\eq{\boldsymbol{\flow}})$  and $\Opt(\game)
	= \txs\min_{\boldsymbol{\flow}\in\flows} \obj(\boldsymbol{\flow})$,
where $\eq{\boldsymbol{\flow}}$ is a \acl{WE} of $\game$.
The game's \acdef{PoA} is then defined as $\PoA(\game)
	= \Eq(\game) /\Opt(\game)$. It is known that Wardrop equilibria need not minimze the social cost, in that case $\PoA(\game) > 1$. For a pair $\pair \in \pairs$, we denote the set of paths that are eventually used in a optimum flow for some demand vector $\boldsymbol{\rate}$ by
	\begin{align*}
\opt{\routes}^\pair = \{\route \in \routes^\pair : \opt{\boldsymbol{\flow}}_\route (\boldsymbol{\rate}) > 0 \text{ for some demand $\boldsymbol{\rate}$ and corresponding social optimum $\opt{\boldsymbol{\flow}}(\boldsymbol{\rate})$}\}.
	\end{align*}
Here and in the following we write $\boldsymbol{\flow}(\boldsymbol{\rate})$ instead of $\boldsymbol{\flow}$ when we want to indicate the corresponding demand vector $\boldsymbol{\rate}$.

\subparagraph{Tolls}
\label{sec:tolls}

We want to explore the possibility of imposing tolls on the edges of the network in such a way that the equilibrium flow of the game with tolls produces a flow that is a solution of the original miminization problem \eqref{eq:SO}. In other words, we want to see whether it is possible to achieve an optimum flow as an equilibrium of a modified game.

We call $\boldsymbol{\toll}=(\toll_{\edge})_{\edge\in\edges} \in \R^\edges$ a \emph{toll vector}. Note that we allow both for positive and negative tolls.
We call $\cost_{\edge}^{\toll}$ the cost of edge $\edge$ under the toll $\tau$, i.e., 
$\cost_{\edge}^{\toll}(\load_{\edge}) := \cost_{\edge}(\load_{\edge})+\toll_{\edge}$.
Similarly $\cost_{\route}^{\toll}(\boldsymbol{\flow}) := \sum_{\edge\in\route}\cost_{\edge}^{\toll}(\load_{\edge})$.
Define
$\game^{\boldsymbol{\toll}} := (\graph,\pairs,\boldsymbol{\cost}^{\boldsymbol{\toll}})$.
A toll vector $\boldsymbol{\toll}$ that for each demand vector $\boldsymbol{\rate} \in \R_+^{\pairs}$ enforces the corresponding system optimum as the equilibrium in $\game^{\boldsymbol{\toll}}$ is called demand-independent optimal toll.

\begin{definition}[Demand-independent optimal toll (DIOT)]
Let $\game = (\graph, \pairs, \boldsymbol{\cost})$. A toll vector $\boldsymbol{\toll} \in \R^\edges$ is called demand-independent optimal toll (DIOT) for $\game$ if for every demand vector $\boldsymbol{\rate} \in \R_+^\pairs$ every corresponding equilibrium with tolls $\boldsymbol{\flow}^{\boldsymbol{\toll}}(\boldsymbol{\rate}) \in \Eq(\game^{\boldsymbol{\toll}})$ is optimal for $\game$, i.e.,
$L(\boldsymbol{\flow}^{\boldsymbol{\toll}}(\boldsymbol{\rate}))
= \sum_{\route \in \route} \flow^{\boldsymbol{\toll}}_\route(\boldsymbol{\rate}) \cost_\route(\boldsymbol{\flow}^{\boldsymbol{\toll}}(\boldsymbol{\rate})) \leq L(\boldsymbol{\flow}(\boldsymbol{\rate})) = \sum_{\route \in \routes} \flow_\route \cost_\route(\boldsymbol{\flow}(\boldsymbol{\rate}))$ for all $\boldsymbol{\flow}(\boldsymbol{\rate}) \in \flows(\boldsymbol{\rate})$.
\end{definition}

In \S~\ref{sec:introduction} we visited two games, the Pigou networks and Braess' paradox, that admit a \ac{DIOT}. The aim of this paper is to characterize the networks $\game = (\graph, \pairs, \boldsymbol{\cost})$ for which \acp{DIOT} exist.


\section{BPR-type cost functions}
\label{sec:monomial}

In this section, we give a complete characterization of the sets of cost functions that admit a \ac{DIOT}. On the positive side, we will show that any network with BPR-type cost functions admits a \ac{DIOT}, independently of the number of commodities and the network topology. On the other hand, we show a strong lower bound proving that for any non-BPR cost function there is a single-commodity game on two parallel edges with costs functions $\cost$ and $\cost + \interc$ for some $\interc \in \R_+$ that does not admit a \ac{DIOT}. Formally, for $\beta > 0$, let
\begin{align*}
\mathcal{C}_{\BPR}(\beta) = \{ \cost : \R_+ \to \R_+ : \cost(\load) = \interc_\cost + a_\cost \load^{\degree} \text{ for all $x \geq 0$, $\slope, \interc \in \R_+$} \}
\end{align*}
be the set of BPR-type cost functions with degree $\beta$.

The following theorem gives a sufficient condition for games with BPR-type cost functions to admit a \ac{DIOT}.

\begin{theorem}\label{thm:toll-vector}
Consider a game $\game = (\graph,\pairs,\boldsymbol{\cost})$ such that there is $\beta \in \R_+$ with $\cost_\edge \in \mathcal{C}_{\BPR}(\beta)$ for all $\edge \in \edges$. Let $\boldsymbol{\toll}$ be a toll vector such that
\begin{equation}\label{eq:TollOptEq}
\sum_{\edge\in\route}\Bigl(  \toll_{\edge}+ \tfrac{\degree}{\degree+1} \interc_{\edge}\Bigr) \leq
\sum_{\edge\in\routealt}\Bigl( \toll_{\edge}+ \tfrac{\degree}{\degree+1} \interc_{\edge}\Bigr).
\end{equation}
for all $\pair \in \pairs$ and all $\route \in \opt{\routes}^\pair$ and all $\routealt \in \routes^\pair$. Then, $\boldsymbol{\toll}$ is a \ac{DIOT}.
\end{theorem}

\begin{proof}
Fix a demand vector $\boldsymbol{\rate}\in\R_{+}^{\pairs}$ and a corresponding optimum flow $\opt{\boldsymbol{\flow}}(\boldsymbol{\rate})$ in $\game$ arbitrarily. We denote by $\opt{\load}_e = \sum_{\route \in \routes : \edge \in \route} \opt{\flow}_\route(\boldsymbol{\rate})$ the load imposed on edge $\edge$ by $\opt{\boldsymbol{\flow}}(\boldsymbol{\rate})$. The local optimality conditions of $\opt{\boldsymbol{\flow}}(\boldsymbol{\rate})$ imply that for all $\pair\in\pairs$ and all $\route,\routealt\in\routes^{\pair}$ with $\opt{\flow}_{\route}(\boldsymbol{\rate})>0$ we have
\begin{align*}
\sum_{\edge\in\route} \cost_{\edge}(\opt{\load}_{\edge}) + \cost_{\edge}'(\opt{\load}_{\edge})\opt{\load}_{\edge} &\leq \sum_{\edge\in\routealt} \cost_{\edge}(\opt{\load}_{\edge}) + \cost_{\edge}'(\opt{\load}_{\edge})\opt{\load}_{\edge} .
\end{align*}
This implies that for all $\pair \in \pairs$ and all $\route, \routealt \in \routes^{\pair}$ with $\opt{\flow}_\route(\boldsymbol{\rate})>0$, there exists a non-negative constant $\const(\route,\routealt) \ge 0$ such that
\begin{align}
\const(\route,\routealt) + \sum_{\edge \in \route} \Bigl( (\degree+1) \slope_{\edge} \opt{\load}_{\edge}^{\degree} + \interc_{\edge} \Bigr) &=
\sum_{\edge\in\routealt}  \Bigl( (\degree+1) \slope_{\edge} \opt{\load}_{\edge}^{\degree} + \interc_{\edge} \Bigr). \label{eq:LocalOpt}
\end{align}
We proceed to show that when the toll vector $\boldsymbol{\toll}$ satisfies \eqref{eq:TollOptEq}, then $\opt{\boldsymbol{\flow}}(\boldsymbol{\rate})$ is a Wardop equilibrium of $\game^{\toll}$. To this end, consider arbitrary $\route,\routealt\in\routes^{\pair}$ with $\opt{\flow}_{\route}(\boldsymbol{\rate})>0$. We have
\begin{align*}
\sum_{\edge \in \route} \Bigl( \cost_{\edge}(\opt{\load}_{\edge}) + \toll_{\edge} \Bigr) &= \sum_{\edge\in\route} \Bigl(\slope_{\edge} \opt{\load}_{\edge}^{\degree} + \interc_{\edge} + \toll_{\edge}\Bigr)\\
&= \sum_{\edge\in\route} \Bigl( (\degree+1) \slope_{\edge} \opt{\load}_{\edge}^{\degree} + \interc_{\edge}\Bigr) - \sum_{\edge\in\route} \Bigl(\degree \slope_{\edge} \opt{\load}_{\edge}^{\degree} - \toll_{\edge} \Bigr)\\
&=
\sum_{\edge\in\routealt}\Bigl((\degree+1)\slope_{\edge} \opt{\load}_{\edge}^{\degree} + \interc_{\edge}\Bigr) - \sum_{\edge\in\route} \Bigl( \degree \slope_{\edge} \opt{\load}_{\edge}^{\degree} - \toll_{\edge}\Bigr) - \const(\route,\routealt),\\
&= \sum_{\edge \in \routealt} \Bigl( \cost_{\edge}(\opt{\load}_{\edge}) + \toll_{\edge} \Bigr) + \sum_{\edge \in \routealt} \Bigl( \degree \slope_{\edge} \opt{\load}_{\edge}^{\degree} - \toll_{\edge} \Bigr) - \sum_{\edge\in\route} \Bigl( \degree \slope_{\edge} \opt{\load}_{\edge}^{\degree} - \toll_{\edge}\Bigr) - \const(\route,\routealt),
\end{align*}
where the third equality comes from \eqref{eq:LocalOpt}. By assumption, the toll vector $\boldsymbol{\toll}$ satisfies equation \eqref{eq:TollOptEq} which implies $\sum_{\edge \in \route} \toll_{\edge} - \sum_{\edge \in \routealt} \toll_{\edge} \leq -\frac{\degree}{\degree+1}\left(\sum_{\edge \in \route} \interc_{\edge} - \sum_{\edge \in \routealt} \interc_{\edge}\right)$. Therefore,
\begin{align*}
\sum_{\edge\in\route} \Bigl(\cost_{\edge}(\opt{\load}_{\edge}) + \toll_{\edge}\Bigr)
&\leq \sum_{\edge\in\routealt}\Bigl(\cost_{\edge}(\opt{\load}_{\edge}) + \toll_{\edge} + \degree \slope_{\edge} \opt{\load}_{\edge}^{\degree} + \tfrac{\degree}{\degree+1}\interc_{\edge}\Bigr) - \sum_{\edge\in\route} \Bigl(\degree\slope_{\edge} \opt{\load}_{\edge}^{\degree} + \tfrac{\degree}{\degree+1}\interc_{\edge}
\Bigr) -\const(\route,\routealt)\\
&=
\sum_{\edge\in\routealt}\Bigl(\cost_{\edge}(\opt{\load}_{\edge}) + \toll_{\edge}\Bigr) + \tfrac{\degree}{\degree+1} \const(\route,\routealt) - \const(\route,\routealt)\\
&= \sum_{\edge\in\routealt}\Bigl(\cost_{\edge}(\opt{\load}_{\edge}) + \toll_{\edge}\Bigr) - \tfrac{1}{\degree+1} \const(\route,\routealt),
\end{align*}
where the first equality stems from \eqref{eq:LocalOpt}.
The statement of the theorem follows from the fact that $\const(\route,\routealt)\ge 0$.
\end{proof}

As an immediate corollary of this theorem, we obtain the existence of \ac{DIOT}s in arbitrary multi-commodity networks with BPR-type cost functions as long as negative tolls are allowed.

\begin{corollary}
Consider a game $\game = (\graph,\pairs,\boldsymbol{\cost})$ with BPR-type cost functions. Then there exists a \ac{DIOT} for $\game$.
\end{corollary}

\begin{proof}
Setting the toll vector $\hat{\boldsymbol{\toll}} = (\hat{\toll}_{\edge})_{\edge\in\edges}$ as
\begin{equation}\label{eq:TrivialOptimalToll}
\hat{\toll}_{\edge} = - \frac{\degree}{\degree+1}\interc_{\edge},
\end{equation}
obviously satisfies \eqref{eq:TollOptEq}, so that the claim follows from \cref{thm:toll-vector}. 
\end{proof}

In the following, we  call $\hat{\boldsymbol{\toll}}$ the \emph{trivial \ac{DIOT}}. 
The necessary condition of \cref{thm:toll-vector} implies in particular that $\sum_{\edge \in \route} \toll_\edge + \frac{\degree}{\degree+1} \interc_\edge$ must be equal for all paths $\route \in \opt{\routes}^\pair$ and all pairs $\pair \in \pairs$. We proceed to show that this is actually a necessary condition for a \ac{DIOT}.

\begin{theorem}\label{pr:ConverseOptimalToll}
Consider a game  $\game = (\graph,\pairs,\boldsymbol{\cost})$ with BPR-type cost functions.
If $\boldsymbol{\toll}$ is a \ac{DIOT} for $\game$, then
\begin{align*}
\sum_{\edge\in\route}\Bigl(  \toll_{\edge}+ \tfrac{\degree}{\degree+1} \interc_{\edge}\Bigr) =
\sum_{\edge\in\routealt}\Bigl( \toll_{\edge}+ \tfrac{\degree}{\degree+1} \interc_{\edge}\Bigr).
\end{align*}
for all $\pair\in\pairs$ and all $\route,\routealt\in \opt{\routes}^{\pair}$.
\end{theorem}

\begin{proof}
Let $\route, \routealt \in \opt{\routes}^\pair$ be arbitrary. By definition of $\opt{\routes}^\pair$ there are demands vectors $\boldsymbol{\rate}, \boldsymbol{\rate}' \in \R_+^{\pairs}$ such that $\opt{\flow}_\route(\boldsymbol{\rate}) > 0$ and $\opt{\flow}_{\routealt}(\boldsymbol{\rate}') > 0$. \cite{hall1978} showed that the path flow functions of a Wardrop equilibrium are continuous functions in the travel demand. For $\lambda \in [0,1]$, let $\boldsymbol{\rate}(\lambda) = (1-\lambda) \boldsymbol{\rate} + \lambda \boldsymbol{\rate'}$ parametrize the travel demands on the convex combination of $\boldsymbol{\rate}$ and $\boldsymbol{\rate'}$. Then, by Hall's result,
there are continuous path flow functions $\flow_\route(\boldsymbol{\rate}(\cdot)) : [0,1] \to \R_+^\edges$ for all $\route \in \routes$ such that $\boldsymbol{\flow}(\boldsymbol{\rate}(\lambda))$ is a Wardrop equilibrium for the travel demand vector $\boldsymbol{\rate}(\lambda)$ for all $\lambda \in [0,1]$. As the system optimal flow $\opt{\boldsymbol{\flow}}$ is a Wardrop equilibrium with respect to the marginal cost function (cf.~\cite{BecMcGWin:Yale1956}), the same holds for the system optimum flow vector $\opt{\boldsymbol{\flow}}$,
i.e., there are continuous path flow functions $\opt{\flow}_\route(\boldsymbol{\rate}(\cdot)) : [0,1] \to \R_+^\edges$ for all $\route \in \routes$ such that $\opt{\boldsymbol{\flow}}(\boldsymbol{\rate}(\lambda))$ is an optimal flow for the travel demand vector $\boldsymbol{\rate}(\lambda)$ for all $\lambda \in [0,1]$. By the continuity of of the path flow functions, there are breakpoints $\lambda_0, \lambda_1, \lambda_2, \dots, \lambda_t \in [0,1]$ such that $0 = \lambda_0 < \lambda_1 < \lambda_2 < \dots < \lambda_t = 1$ and the support $\mathcal{S}(\opt{\boldsymbol{f}}(\boldsymbol{\rate}(\lambda)) = \{\route \in \routes : \opt{\flow}_\route(\boldsymbol{\rate}(\lambda)) > 0\}$ is constant on $(\lambda_j, \lambda_{j+1})$ for all $j \in \{0,\dots,t-1\}$.
We claim that
\begin{align*}
\sum_{\edge \in \route} \Bigl(\toll_\edge + \tfrac{\degree}{\degree+1} \interc_\edge \Bigr) = \sum_{\edge \in \routealt} \Bigl(\toll_\edge + \tfrac{\degree}{\degree+1} \interc_\edge \Bigr) \text{ for all }\route,\routealt \in \bigcup_{j \in \{0,\dots,t-1\}} \mathcal{S}\bigl(\opt{\boldsymbol{f}}\bigl(\boldsymbol{\rate}\bigl(\tfrac{\lambda_j + \lambda_{j+1}}{2}\bigr)\bigr)\bigr).
\end{align*}
Since, by construction, $\route \in \mathcal{S}\bigl(\opt{\boldsymbol{f}}\bigl(\boldsymbol{\rate}\bigl(\tfrac{\lambda_0 + \lambda_{1}}{2}\bigr)\bigr)\bigr)$ and $\routealt \in \mathcal{S}\bigl(\opt{\boldsymbol{f}}\bigl(\boldsymbol{\rate}\bigl(\tfrac{\lambda_{t-1} + \lambda_{t}}{2}\bigr)\bigr)\bigr)$, the claim implies the result.

We proceed to prove the claim by induction on the number of breakpoints. First, let $j \in \{0,\dots,t-1\}$ be arbitrary and consider a travel demand $\boldsymbol{\rate}_{j,j+1} = \boldsymbol{\rate}\bigl(\tfrac{\lambda_j + \lambda_{j+1}}{2}\bigr)$. We denote by $\opt{\load}_\edge(\boldsymbol{\rate}_{j,j+1}) = \sum_{\route \in \routes : \edge \in \route} \opt{\flow}_\route(\boldsymbol{\rate}_{j,j+1}))$ the load imposed by the corresponding optimal flow.

Since  $\boldsymbol{\toll}$ is a \ac{DIOT} for $\game$, the optimal flow $\opt{\boldsymbol{\flow}}(\boldsymbol{\rate}_{j,j+1})$ is a Wardrop equilibrium with respect to $\boldsymbol{\toll}$, i.e., $\sum_{\edge \in \route} (\cost_{\edge}(\opt{\load}_{\edge}(\boldsymbol{\rate}_{j,j+1})) + \toll_{\edge}) = \sum_{\edge \in \routealt} (\cost_{\edge}(\opt{\load}_{\edge}(\boldsymbol{\rate}_{j,j+1})) + \toll_{\edge})$ for all $\route,\routealt \in \mathcal{S}(\opt{\boldsymbol{\flow}}(\boldsymbol{\rate}_{j,j+1}))$ which implies
\begin{equation}
\sum_{\edge\in\route}\left(\slope_{\edge} \opt{\load}_{\edge}(\boldsymbol{\rate}_{j,j+1})^{\degree} + \interc_{\edge} + \toll_{\edge}\right)
=
\sum_{\edge\in\routealt}\left(\slope_{\edge} \opt{\load}(\boldsymbol{\rate}_{j,j+1})_{\edge}^{\degree} + \interc_{\edge} + \toll_{\edge}\right). \label{eq:necessary_1}
\end{equation}
for all $\route,\routealt \in \mathcal{S}(\opt{\boldsymbol{\flow}}(\boldsymbol{\rate}_{j,j+1}))$. By the local optimality conditions of $\opt{\boldsymbol{\flow}}(\boldsymbol{\rate}_{j,j+1}))$, we further have $$\sum_{\edge \in \route} \cost_{\edge}(\opt{\load}_{\edge}(\boldsymbol{\rate}_{j,j+1})) + \cost_{\edge}'(\opt{\load}_{\edge}(\boldsymbol{\rate}_{j,j+1})) \opt{\load}_{\edge}(\boldsymbol{\rate}_{j,j+1}) = \sum_{\edge \in \routealt} \cost_{\edge}(\opt{\load}_{\edge}(\boldsymbol{\rate}_{j,j+1})) + \cost_{\edge}'(\opt{\load}_{\edge}(\boldsymbol{\rate}_{j,j+1})) \opt{\load}_{\edge}(\boldsymbol{\rate}_{j,j+1}),$$
which is equivalent to 
\begin{equation}
\sum_{\edge\in\route}\Bigl((\degree+1)\slope_{\edge} \opt{\load}_{\edge}(\boldsymbol{\rate}_{j,j+1})^{\degree} + \interc_{\edge} \Bigr)
=
\sum_{\edge\in\routealt}\Bigl((\degree+1)\slope_{\edge} \opt{\load}_{\edge}(\boldsymbol{\rate}_{j,j+1})^{\degree} + \interc_{\edge} \Bigr) \label{eq:necessary_2}
\end{equation}
for all $\route,\routealt \in \mathcal{S}(\opt{\boldsymbol{\flow}}(\boldsymbol{\rate}_{j,j+1}))$.
Subtracting \eqref{eq:necessary_2} from $(\degree+1)$ times \eqref{eq:necessary_1} and dividing by $\degree+1$ we obtain
\begin{equation}
\label{eq:equality}
\sum_{\edge\in\route}\left( \tfrac{\degree}{\degree+1} \interc_{\edge} + \toll_{\edge}\right)
=
\sum_{\edge\in\routealt}\left( \tfrac{\degree}{\degree+1}\interc_{\edge} +  \toll_{\edge}\right)
\end{equation}
for all $j \in \{0,\dots,t-1\}$ and all $\route,\routealt \in \mathcal{S}(\opt{\boldsymbol{\flow}}(\boldsymbol{\rate}_{j,j+1}))$. To finish the proof, note that \eqref{eq:necessary_1} and \eqref{eq:necessary_2} imply that for all $j \in \{0,\dots,t-1\}$ there exist functions $f : [0,1] \to \R_+$ and $g : [0,1] \to \R_+$ such that
\begin{align*}
f_{j,j+1}(\lambda) &= \sum_{\edge \in \route} \slope_\edge \opt{\load}_\edge(\boldsymbol{\rate}(\lambda))^\degree	+ \interc_\edge + \toll_\edge\\
g_{j,j+1}(\lambda) &= \sum_{\edge \in \route} (\degree+1)\slope_\edge \opt{\load}_\edge(\boldsymbol{\rate}(\lambda))^\degree	+ \interc_\edge 
\end{align*}
for all $\route \in \mathcal{S}(\opt{\boldsymbol{\flow}}(\boldsymbol{\rate}_{j,j+1}))$ and all $\lambda \in (\lambda_j, \lambda_{j+1})$. By standard continuity arguments, these equations hold also for all $\lambda \in [\lambda_j, \lambda_{j+1}]$. In particular, we obtain for $\boldsymbol{\rate}_j = \boldsymbol{\rate}(\lambda_j)$ the equations
\begin{align*}
\sum_{\edge\in\route}\Bigl((\degree+1)\slope_{\edge} \opt{\load}_{\edge}(\boldsymbol{\rate}_{j,j+1})^{\degree} + \interc_{\edge} \Bigr)
&=
\sum_{\edge\in\routealt}\Bigl((\degree+1)\slope_{\edge} \opt{\load}_{\edge}(\boldsymbol{\rate}_{j,j+1})^{\degree} + \interc_{\edge} \Bigr)	\\
\sum_{\edge\in\route}\left( \tfrac{\degree}{\degree+1} \interc_{\edge} + \toll_{\edge}\right)
&=
\sum_{\edge\in\routealt}\left( \tfrac{\degree}{\degree+1}\interc_{\edge} +  \toll_{\edge}\right)
\end{align*}
for all $\route,\routealt \in \mathcal{S}(\opt{\boldsymbol{\flow}}(\boldsymbol{\rate}_{j-1,j})) \cup \mathcal{S}(\opt{\boldsymbol{\flow}}(\boldsymbol{\rate}_{j,j+1}))$, and the claim follows from the same arguments as before.
\end{proof}

We proceed to show that the set of BPR-type cost functions
is the largest set of cost functions that  guarantee the existence of a \ac{DIOT}, even for single-commodity networks consisting of two parallel edges. The proof is deferred to the appendix.

\begin{theorem}
\label{thm:characterization}
Let $c$ be twice continuously differentiable, strictly semi-convex and strictly increasing, but not of BPR-type. Then there is a congestion game $\game = (\graph,\pairs,\boldsymbol{\cost})$ with two parallel edges and cost functions $\cost(\load)$ and $\cost(\load) + \interc$ for some $\interc \in \R_{+}$ that does not have a \ac{DIOT}.
\end{theorem}

A similar construction as in the proof of Theorem~\ref{thm:characterization} shows also that a network with two parallel edges with cost functions $c_1 \in \mathcal{C}_{\BPR}(\beta_1)$ and $c_2 \in \mathcal{C}_{\BPR}(\beta_2)$ does not admit a \ac{DIOT} if $\beta_1 \neq \beta_2$.

\section{Nonnegative tolls}
\label{sec:nonnegative}

The trivial \ac{DIOT} toll $\hat{\boldsymbol{\toll}}$ is the trivial solution for both the sufficient condition for a \ac{DIOT} imposed by Theorem~\ref{thm:toll-vector} and the necessary condistion for a \ac{DIOT} shown in Theorem~\ref{pr:ConverseOptimalToll}. However, the trivial \ac{DIOT} is always negative so that the system designer needs to subsidize the traffic in order to enforce the optimum flow. One may wonder whether the conditions imposed by Theorem~\ref{thm:toll-vector} and Theorem~\ref{pr:ConverseOptimalToll} admit also a non-negative solution. Our next result shows that, for games played on a \acdef{DAMG}, a non-negative \ac{DIOT} can always be found. The proof is deferred to the appendix.

\begin{theorem}\label{thm:OptimalNonnegativeToll}
Consider a game $\game = (\graph,\pairs,\boldsymbol{\cost})$ with BPR-type cost functions where $\graph$ is a \ac{DAMG}. Then there exists a non-negative \ac{DIOT}  for $\game$.
\end{theorem}


\begin{figure*}[tb]
\centering
\footnotesize

\begin{tikzpicture}[scale=1.1,edgestyle/.style={-},>=stealth]
\tikzstyle{every circle node}=[fill,inner sep=1.5pt,outer sep=0pt]

\coordinate (A) at (-\textwidth/6,0);
\coordinate (B) at (\textwidth/6,0);

\draw (A) node(A)[circle] {} node[left]{$\vertex\vphantom{bp}$};
\draw (B) node(B)[circle] {} node[right]{$\vertexalt\vphantom{bp}$};

\draw [edgestyle,->] (A) to [bend left = 60] node [midway, above] {$\cost_{1}(\load) = \load$} (B);
\draw [edgestyle,->] (A) to [bend left = 30]  node [midway, below] {$\cost_{2}(\load) = 1$} (B);
\draw [edgestyle,->] (B) to [bend left = 30] node [midway, above] {$\cost_{3}(\load) = \load$} (A);
\draw [edgestyle,->] (B) to [bend left = 60] node [midway, below] {$\cost_{4}(\load) = 1$} (A);
\end{tikzpicture}

\caption{A cyclic network with positive tolls.}
\label{fig:cyclic}
\end{figure*}
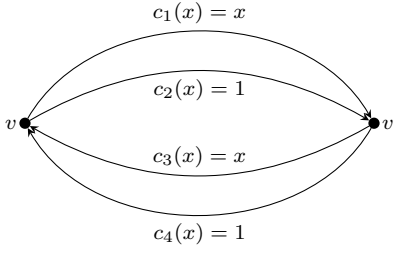


The condition that the graph $\graph$ is a \ac{DAMG} is sufficient for the existence of a non-negative \ac{DIOT}. It is not necessary, as the following counterexample  shows. 

\begin{example}
\label{ex:double-pigou}
Let $\game = (\graph,\pairs,\boldsymbol{\cost})$ with $\pairs=\braces{1,2}$, $\vertices=\braces{\vertex,\vertexalt}$,
$\source_{1}=\sink_{2}=\vertex$, $\source_{2}=\sink_{1}=\vertexalt$,
$\edge_{1},\edge_{2}\in (\source_{1}\to\sink_{1})$, 
$\edge_{3},\edge_{4}\in (\source_{2}\to\sink_{2})$, and the costs are as in \cref{fig:cyclic} (a).
The graph $\graph$ is not a \ac{DAMG}, but the following non-negative toll is a \ac{DIOT}:
\begin{align*}
\toll_{1}&= \frac{1}{2} &  \toll_{2}&=0, & \toll_{3}&= \frac{1}{2}, & \toll_{4}=0.
\end{align*}
\end{example}

We proceed to show that for graphs that contain a directed cycle, non-negative \acp{DIOT} need not exist, even in networks with affine costs.

\begin{proposition}
There are networks with affine costs that do not admit a non-negative \ac{DIOT}.
\end{proposition}

\begin{proof}
Consider the game $\game = (\graph, \pairs, \boldsymbol{\cost})$ with $\pairs = \{1\}$, $\vertices = \{\source,u,v,\sink\}$, and $e_1 \in (\source \to u)$, $e_2 \in (u \to \sink)$, $e_3 \in (v \to u)$, $e_4 \in (u \to v)$, $e_5 \in (\source \to v)$, and $e_6 \in (v,\sink)$ shown in \cref{fig:cost}. For an edge $e_j$ we denote by $\cost_{j}$ the corresponding cost function. In particular we assume
\begin{align*}
\cost_{1}(\load) &= \cost_{4}(\load) = \cost_{6}(\load) = 2, &
\cost_{3}(\load) &= 1, &
\cost_{2}(\load) &= \cost_{5}(\load) = 4x.
\end{align*}

%
%


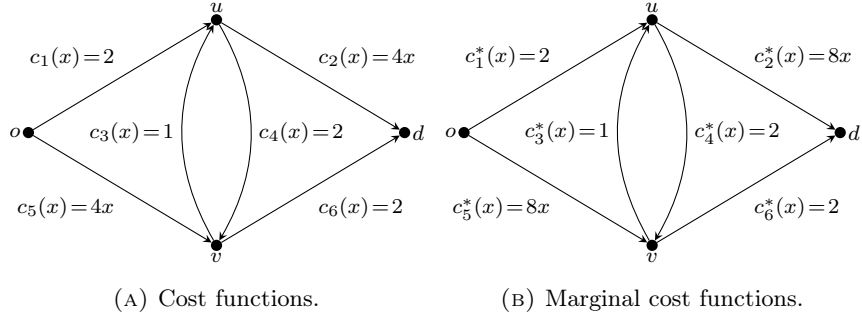
\begin{figure*}[tb]
\centering
\footnotesize

\begin{subfigure}[b]{0.45\textwidth}
\centering
\begin{tikzpicture}[scale=1.25,edgestyle/.style={-},
>=stealth]
\tikzstyle{every circle node}=[fill,inner sep=1.5pt,outer sep=0pt]
\coordinate (A) at (-2,0);
\coordinate (B) at (0,1.2);
\coordinate (C) at (0,-1.2);
\coordinate (D) at (2,0);

\draw (A) node[circle](A) {} node[left] {$\source$};
\draw (B) node[circle](B) {} node[above] {$u$};
\draw (C) node[circle](C) {} node[below] {$v$};
\draw (D) node[circle](D) {} node[right] {$\sink$};

\draw [edgestyle,->] (A) to node [midway, above left] {$\cost_{1}(\load) \!=\! 2$} (B);
\draw [edgestyle,->] (B) to [bend left = 30] node [midway, right] {$\cost_{4}(\load) \!=\! 2$} (C);
\draw [edgestyle,->] (C) to node [midway, below right] {$\cost_{6}(\load) \!=\! 2$} (D);
\draw [edgestyle,->] (A) to node [midway, below left] {$\cost_{5}(\load) \!=\! 4\load$} (C);
\draw [edgestyle,->] (B) to node [midway, above right] {$\cost_{2}(\load) \!=\! 4\load$} (D);
\draw [edgestyle,->] (C) to [bend left = 30] node [midway, left] {$\cost_{3}(\load) \!=\! 1$} (B);
\end{tikzpicture}
\subcaption{Cost functions.\label{fig:pigou}}
\end{subfigure}
\begin{subfigure}[b]{0.45\textwidth}
\centering
\begin{tikzpicture}[scale=1.25,edgestyle/.style={-},
>=stealth]
\tikzstyle{every circle node}=[fill,inner sep=1.5pt,outer sep=0pt]
\coordinate (A) at (-2,0);
\coordinate (B) at (0,1.2);
\coordinate (C) at (0,-1.2);
\coordinate (D) at (2,0);

\draw (A) node[circle](A) {} node[left] {$\source$};
\draw (B) node[circle](B) {} node[above] {$u$};
\draw (C) node[circle](C) {} node[below] {$v$};
\draw (D) node[circle](D) {} node[right] {$\sink$};

\draw [edgestyle,->] (A) to node [midway, above left] {$\cost^*_{1}(\load) \!=\! 2$} (B);
\draw [edgestyle,->] (B) to [bend left = 30] node [midway, right] {$\cost^*_{4}(\load) \!=\! 2$} (C);
\draw [edgestyle,->] (C) to node [midway, below right] {$\cost^*_{6}(\load) \!=\! 2$} (D);
\draw [edgestyle,->] (A) to node [midway, below left] {$\cost^*_{5}(\load) \!=\! 8\load$} (C);
\draw [edgestyle,->] (B) to node [midway, above right] {$\cost^*_{2}(\load) \!=\! 8\load$} (D);
\draw [edgestyle,->] (C) to [bend left = 30] node [midway, left] {$\cost^*_{3}(\load) \!=\! 1$} (B);
\end{tikzpicture}
\subcaption{Marginal cost functions.}
\end{subfigure}

\caption{Cost functions $\cost$ and marginal cost functions $\cost^*$ for a cyclic network.\label{fig:cost}}
\end{figure*}


We proceed to describe the pattern of the socially optimum flow $\opt{\boldsymbol{\flow}}(\rate)$ as a function of the demand $\rate \geq 0$. To verify this, it may be helpful to consult the marginal cost functions $c^*(x) = (x c(x))'$ shown in \cref{fig:cost} (b) since the system optimum flow $\opt{\boldsymbol{\flow}}(\rate)$ is a Wardrop equilibrium for $c^*$

For demands $\rate \le 1/8$, only the path $(\edge_{\source,v},\edge_{v,u},\edge_{u,\sink})$ is used. When $1/8<\rate< 1/4$ the three paths $(\edge_{\source,v},\edge_{v,u},\edge_{u,\sink})$, $(\edge_{\source,u},\edge_{u,\sink})$, and $(\edge_{\source,v},\edge_{v,\sink})$ are used in the system optimum. Finally for $\rate > 1/4$ the paths $(\edge_{\source,u},\edge_{u,v},\edge_{v,\sink})$, $(\edge_{\source,u},\edge_{v,\sink})$, and $(\edge_{\source,v},\edge_{v,\sink})$ are used in the system optimum. Note that all paths four paths from $\source$ to $\sink)$ are used in the system optimum for some demand so that $\opt{\routes}^1 = \routes^1$.

Theorem~\ref{pr:ConverseOptimalToll} implies that for any \ac{DIOT} $\boldsymbol{\toll}$ we have $\sum_{\edge \in \route} \interc_\edge/2 + \toll_\edge = \sum_{\edge \in \routealt} \interc_\edge/2 + \toll_\edge$ for all all $\route, \routealt \in \opt{\routes}^i$. We derive the existence of a constant $T \in \R$ such that the following equations hold:
\begin{subequations}
\begin{align}
T &= \frac{1}{2} + \toll_5 + \toll_3 + \toll_2,\label{eq:i} \\
T &= 3 + \toll_1 + \toll_4 + \toll_6,\label{eq:ii}\\
T &= 1 + \toll_1 + \toll_2,\label{eq:iii}\\
T &= 1 + \toll_5 + \toll_6.\label{eq:iv}	
\end{align}
Subtracting both \eqref{eq:iii} and \eqref{eq:iv} from \eqref{eq:i} + \eqref{eq:i}, we obtain $0 = \frac{3}{2} + \toll_3 + \toll_4$ which implies that either $\toll_3$ or $\toll_4$ must be negative. 
\end{subequations}
%
%
\end{proof}

\section{Aggregatively non-negative Tolls}
\label{sec:general}

When nonnegative \acp{DIOT} do not exist, it is conceivable that a social planner may sometime want to use negative tolls in order to achieve her goal. 
Nevertheless, the planner may be subject to budget constraints and not be able to afford a toll system that implies a global loss. 
Therefore it is interesting to study the existence of conditions for a \ac{DIOT} $\boldsymbol{\toll}$ such that the following budget constraint is satisfied:
\begin{equation}\label{eq:Budget}
\sum_{\edge\in\edges}\toll_{\edge}\load_{\edge} \ge 0, \quad\text{for any feasible flow }\boldsymbol{\flow}.
\end{equation}
Intuitively, \eqref{eq:Budget} requires that the social planner does not loose money for any feasible flow.
In this section, we show that when the origin-destination pairs $(\source^i,\sink^i)$ satisfy a order condition, then a \ac{DIOT} satisfying the budget constraint exists.

\begin{theorem}\label{thm:OrderedMultipleOD}
Consider a game $\game = (\graph,\pairs,\boldsymbol{\cost})$ with BPR-type cost functions. 
If there exists an order $\prec$ on $\vertices$ such that for all $\pair\in\pairs$ we have $\source^{\pair}\prec\sink^{\pair}$, then there exists a \ac{DIOT} $\boldsymbol{\toll}$ that satisfies \eqref{eq:Budget}.
\end{theorem}

We obtain the existence the budget feasible \acp{DIOT} for single commodity networks as a direct corollary of \ref{thm:OrderedMultipleOD}.

\begin{corollary}\label{co:BudgetConstraintSingleOD}
Consider a game $\game = (\graph,\pairs,\boldsymbol{\cost})$ with BPR-type cost and a single \ac{OD} pair. Then, there exists a \ac{DIOT} $\boldsymbol{\toll}$ that satisfies \eqref{eq:Budget}.
\end{corollary}

Example~\ref{ex:double-pigou} shows that the condition of \cref{thm:OrderedMultipleOD} is only sufficient for the existence of a \ac{DIOT} that satisfies \eqref{eq:Budget}.


\subparagraph*{Acknowledgements}
~\\
Riccardo Colini-Baldeschi and Marco Scarsini are members of GNAMPA-INdAM.
Max Klimm gratefully acknowledges the support and hospitality of LUISS during a visit in which this research was initiated.


\bibliographystyle{apalike}
\bibliography{../bibtex/biblimitpoa}

\newpage
\appendix

\section{Missing Proofs of Section~\ref{sec:monomial}}

\subsection{Proof of Theorem~\ref{thm:characterization}}

\begin{proof}
Let $\alpha \in \R_+$ be arbitrary. All solutions of the differential equation
\begin{align}
\label{eq:diff_equation}
\frac{\load f''(\load)}{f'(\load)} = \alpha	
\end{align}
are of the form $f(\load) = a \frac{\load^{\alpha+1}}{\alpha+1} +b$ with $a,b \in \R$. Since $c$ is not of BPR-type, $c$ does not solve \eqref{eq:diff_equation}. In particular, there are $y,z \in \R_+$, with $0 < y < z$ such that
\begin{align}
\label{eq:xy}
\frac{y c''(y)}{c'(y)} \neq \frac{z c''(z)}{c'(z)}.
\end{align}
By continuity, we can choose $y > 0$ such that there is $\varepsilon > 0$ so that \eqref{eq:xy} hold for all $z \in (y,y+\varepsilon)$. 

Consider the game $\game = (\graph, \pairs, \boldsymbol{\cost})$ with $\vertices = \{\source,\sink\}$ and $e_1, e_2 \in (\source \to \sink)$ shown in \cref{fig:only_monomial}. The upper edge $e_1$ has cost function $c(x)$ while the lower edge has cost function $c(x) + \interc$ for some constant $\interc \in \R_{+}$ to be fixed later. Since the edges are parallel, without loss of generality, we can consider toll vectors $\boldsymbol{\toll}=(\toll,0)$, with  $\toll\geq 0$. 

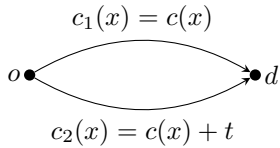
\begin{figure}[b]
\begin{tikzpicture}[edgestyle/.style={-},>=stealth]
\tikzstyle{every circle node}=[fill,inner sep=1.5pt,outer sep=0pt]
\draw (0,0) node(s)[circle]{} node[left]{$\source$};
\draw (3,0) node(t)[circle]{} node[right]{$\sink$};
\draw[->] (s) to[bend left]  node[pos=0.5,above] {$\cost_1(\load) = c(\load)$} (t);
\draw[->] (s) to[bend right] node[pos=0.5,below] {$\cost_2(\load) = c(\load)+\interc$} (t);
\end{tikzpicture}
\caption{\label{fig:only_monomial} Network consisting of two parallel edges.}
\end{figure}

For small flow demands $\rate$, both the system optimum and the equilibrium use the upper edge only. The system optimum starts using the lower link for the unique $\rate_0$ satisfying
\begin{align*}
c(\rate_0) + \rate_0 c'(\rate_0) = \cost(0) + \interc.	
\end{align*}
To see that $\rate_0$ is unique, consider the marginal cost function $\cost^* : \R_+ \to [c(0),\infty)$ defined as $\cost^*(\load) = \cost(\load) + \load \cost'(\load)$. As $c$ is strictly semi-convex, $xc(x)$ is convex and, thus, $\cost^*(\load) = (\load\cost(\load))'$ is strictly increasing. This implies that the inverse function $(\load^*)^{-1} : [\cost(0),\infty) \to \R_+$ is well-defined. We can then express $\rate_0$ as $\rate_0 = (\cost^*)^{-1}(\cost(0)+ \interc)$ so that $\rate_0$ is unique.

For $\boldsymbol{\toll}$ to be a \ac{DIOT} it is necessary that also the equilibrium flow starts using the lower edge at demand $\rate_0$, i.e., it is necessary that
\begin{align*}
c(\rate_0) + \tau = c(0) + \interc. 
\end{align*}
As a consequence, for all demands $\rate \geq \rate_0$ both edges are used in equilibrium, and the equilibrium flow $(\load_1(\rate),\load_2(\rate))$ is governed by the system of differential equations
\begin{align*}
\load_1'(\rate) &= \frac{\frac{1}{c'(\load_1(\rate))}}{\frac{1}{c'(\load_1(\rate))} + \frac{1}{c'(\load_2(\rate))}}	&
\load_2'(\rate) &= \frac{\frac{1}{c'(\load_2(\rate))}}{\frac{1}{c'(\load_1(\rate))} + \frac{1}{c'(\load_2(\rate))}}
\end{align*}
with the initial values $x_1(\rate_0) = \rate_0$ and $x_2(\rate_0) = 0$, see \cite{HarKleKliMoe:Network2015} for a reference. Similarly, since the system optimum is at equilibrium with respect to the marginal costs $\cost^*$, its flow $(\opt{\load}_1(\rate), \opt{\load}_2(\rate))$ is governed by the system of differential equations
\begin{align*}
\opt{\load}_1'(\rate) &= \frac{\frac{1}{2c'(\opt{\load}_1(\rate))+\opt{\load}_1 c''(\opt{\load}_1(\rate)) }}{\frac{1}{2c'(\opt{\load}_1(\rate))+\opt{\load}_1(\rate) c''(\opt{\load}_1(\rate)) } + \frac{1}{2c'(\opt{\load}_2)+\opt{\load}_2(\rate) c''(\opt{\load}_2(\rate)) }}	\\
\opt{\load}_2'(\rate) &= \frac{\frac{1}{2c'(\opt{\load}_2(\rate))+\opt{\load}_2(\rate) c''(\opt{\load}_2(\rate)) }}{\frac{1}{2c'(\opt{\load}_1(\rate))+\opt{\load}_1 c''(\opt{\load}_1(\rate)) } + \frac{1}{2c'(\opt{\load}_2(\rate))+\opt{\load}_2 c''(\opt{\load}_2(\rate)) }}
\end{align*}
with the initial values $\opt{\load}_1(\rate_0) = \rate_0$ and $\opt{\load}_2(\rate_0) = 0$. For $\tau$ to be a DIOT it is necessary to have that $\load_2(\rate_0) = \opt{\load}_2(\rate_0)$ for all $\rate \geq \rate_0$.

We claim that we can choose $\interc > 0$ such that no \ac{DIOT} exists.
For a contradiction, let us assume that for all $\interc > 0$ we have $\load_2(\rate) = \opt{\load}_2(\rate)$ for all $\rate \geq \rate_0$. 
This implies in particular that $\load_2'(\rate_0) = \opt{\load}_2'(\rate_0)$ for all $\rate \geq \rate_0$. Since we assumed $\opt{\load}_1(\rate) = \load_1(\rate)$ and $\opt{\load}_2(\rate) = \load_2(\rate)$ for all $\rate \geq \rate_0$, this implies
\begin{align*}
\frac{\frac{1}{c'(\load_2(\rate))}}{\frac{1}{c'(\load_1(\rate))} + \frac{1}{c'(\load_2(\rate))}} &= \frac{\frac{1}{2c'(\load_2(\rate))+\load_2(\rate) c''(\load_2(\rate)) }}{\frac{1}{2c'(\load_1(\rate))+\load_1(\rate) c''(\load_1(\rate)) } + \frac{1}{2c'(\load_2(\rate))+\load_2(\rate) c''(\load_2(\rate))) }} & &\text{ for all $\rate \geq \rate_0$.}
\intertext{This is equivalent to}
\frac{1}{\frac{c'(\load_2(\rate))}{c'(\load_1(\rate))} + 1} &= \frac{1}{\frac{2c'(\load_2(\rate)) + \load_2(\rate)c''(\load_2(\rate))}{2c'(\load_1(\rate))+\load_1(\rate) c''(\load_1(\rate)) } + 1} & &\text{ for all $\rate \geq \rate_0$}
\intertext{implying}
\frac{c'(\load_2(\rate))}{c'(\load_1(\rate))} &= \frac{2c'(\load_2(\rate))+\load_2(\rate) c''(\load_2(\rate))}{2c'(\load_1(\rate)) + \load_1(\rate)c''(\load_1(\rate))} & &\text{ for all $\rate \geq \rate_0$.}
\end{align*}
We obtain
\begin{align}
\label{eq:xy_new}
\frac{x_1(\rate)c''(x_1(\rate))}{c'(x_1(\rate))} &= \frac{x_2(\rate)c''(x_2(\rate))}{c'(x_2(\rate))}  & &\text{ for all $\rate \geq \rate_0$.}
\end{align}
Since $\cost$ is strictly increasing, both $\load_1(\rate)$ and $\load_2(\rate)$ tend to infinity as $\rate$ goes to infinity. Since, in addition, $\load_2(\rate_0) = 0$ and $\load_2(\rate)$ is continuous, there is $\rate^*$ so that $\load_2(\rate^*) = y$. Such $\rate^*_t = \rate^*$, in fact, exists for all $\interc >0$. We claim that $\lim_{\interc \to 0} x_1(\rate^*_t) = y$. To see this, note that we have $c^*(x_1(\rate_t^*)) = c^*(x_2(\rate_t^*))$ since we assumed that the equilibrium flow $\load$ is optimal. This implies
\begin{align*}
\cost(\load_1(\rate_t^*)) + \load_1(\rate_t^*) \cost'(\load_1(\rate_t^*)) &= 	\cost(\load_2(\rate_t^*)) + \interc + \load_2(\rate_t^*) c'(\load_2(\rate_t^*))\\
&= c(y) + t + yc'(y),
\end{align*}
so that $x_1(\rate_t^*) = (\cost^*)^{-1}(c^*(y) + t)$. As $c^*$ is strictly increasing, $(c^*)^{-1}$ is continuous implying $\lim_{\interc \to 0} x_1(\rate_t^*) = y$. From there, we derive that for $\interc$ sufficiently small we have $x_1(\rate^*_t) \in (y,y+\varepsilon)$. We constructed a contradiction between \eqref{eq:xy} and \eqref{eq:xy_new}, which finishes the proof.
 \end{proof}

\subsection{Proof of Theorem~\ref{thm:OptimalNonnegativeToll}}

\begin{proof}
Given a \ac{DAMG} there exists a topological sort, namely a linear ordering $\prec$ of its vertices such that, if $\vertex\prec\vertexalt$, then there is no path from $\vertexalt$ to $\vertex$ in the \ac{DAMG}. Notice that, in general the topological sort of a \ac{DAMG} is not unique. Let $\abs{\vertices}=n$ and call $\boldsymbol{\vertex}^{\prec} = (\vertex_{(1)}, \dots, \vertex_{(n)})$ the vector of ordered vertices. For each edge $\edge\in(\vertex_{(i)}\to\vertex_{(j)})$, define 
\begin{equation}\label{eq:Difference}
\delta_{\edge}:=j-i.
\end{equation} 
Let $\hat{\boldsymbol{\toll}}$ be the trivial \ac{DIOT} of the game $\game$ and let
\begin{equation}\label{eq:xi}
\xi:=\min_{\edge\in\edges}\frac{\hat{\toll}_{\edge}}{\delta_{\edge}}
\quad\text{and}\quad
\chi:=\xi_{-},
\end{equation}
where $\xi_{-}=\max\{-\xi, 0\}$ is the negative part of $\xi$. 
Define now
\begin{equation}\label{eq:NonnegativeToll}
\toll_{\edge} = \hat{\toll}_{\edge} + \delta_{\edge}\chi.
\end{equation}
We first prove that the toll vector $\boldsymbol{\toll}$ is non-negative. 
Notice that $\chi$, defined as in \eqref{eq:xi} is non-negative and  $\chi=0$ only if $\hat{\toll}_{\edge} \ge 0$ for all $\edge\in\edges$.
Assume that there exists a $\hat{\toll}_{\edge}<0$ and let $\eq{\edge}\in\argmin_{\edge\in\edges}\hat{\toll}_{\edge}/\delta_{\edge}$. Then
\begin{equation*}
\toll_{\eq{\edge}}=\hat{\toll}_{\eq{\edge}}+\delta_{\eq{\edge}}\chi
=\hat{\toll}_{\eq{\edge}}-\delta_{\eq{\edge}}\frac{\hat{\toll_{\eq{\edge}}}}{\delta_{\eq{\edge}}}=0.
\end{equation*}
In general, whenever $\toll_{\edge} <0$, we have
\begin{equation*}
\toll_{\edge}=\hat{\toll}_{\edge}+\delta_{\edge}\chi
=\hat{\toll}_{\edge}-\delta_{\edge}\frac{\hat{\toll_{\eq{\edge}}}}{\delta_{\eq{\edge}}}
\ge \hat{\toll}_{\edge}-\delta_{\edge}\frac{\hat{\toll_{\edge}}}{\delta_{\edge}}
=0.
\end{equation*}

Now we prove that the toll vector $\boldsymbol{\toll}$ is a \ac{DIOT}. By \cref{thm:toll-vector}, this means that it satisfies equation \eqref{eq:TollOptEq}.
First, notice that, by construction of the $\delta_{\edge}$, for any $\pair\in\pairs$, we have
\begin{equation}\label{eq:sumdelta}
\sum_{\edge\in\route} \delta_{\edge} = \sum_{\edge\in\routealt} \delta_{\edge}\quad\text{for all } \route,\routealt \in \routes^{\pair}.
\end{equation}
By \eqref{eq:TollOptEq} we have
\begin{equation*}
\sum_{\edge\in\route}\Bigl( (\degree+1) \hat{\toll}_{\edge}+ \degree \interc_{\edge}\Bigr)=
\sum_{\edge\in\routealt}\Bigl( (\degree+1) \hat{\toll}_{\edge}+ \degree \interc_{\edge}\Bigr),
\end{equation*} 
hence, by \eqref{eq:sumdelta},
\begin{equation*}
\sum_{\edge\in\route}\Bigl( (\degree+1)(\hat{\toll}_{\edge} + \delta_{\edge}\chi)+ \degree\interc_{\edge}\Bigr)=
\sum_{\edge\in\routealt}\Bigl( (\degree+1)(\hat{\toll}_{\edge} + \delta_{\edge}\chi)+ \degree\interc_{\edge}\Bigr),
\end{equation*} 
that is
\begin{equation*}
\sum_{\edge\in\route} \left((\degree+1)\toll_{\edge} + \degree\interc_{\edge}\right)=
\sum_{\edge\in\routealt} \left((\degree+1)\toll_{\edge} + \degree\interc_{\edge}\right),
\end{equation*}
which finishes the proof.
\end{proof}

\section{Missing Proofs of Section~\ref{sec:general}}

\subsection{Proof of Theorem~\ref{thm:OrderedMultipleOD}}

\begin{proof}
Consider the trivial \ac{DIOT} $\hat{\boldsymbol{\toll}}$. 
Define now
\begin{equation}\label{eq:NonnegativeTollgamma}
\toll_{\edge}(\increm) = \hat{\toll}_{\edge} + \delta_{\edge}\increm,
\end{equation}
where $\delta$ is defined as in \eqref{eq:Difference}.
Notice that $\delta_{\edge}$ can be negative, but, for each path $\route\in\routes^{\pair}$ we have
\begin{equation*}
\sum_{\edge\in\route}\delta_{\edge}=\sink^{\pair}-\source^{\pair}=:\diffOD^{\pair}.
\end{equation*}
Hence
\begin{align*}
\sum_{\edge\in\edges}\toll_{\edge}(\increm)\load_{\edge} 
&=
\sum_{\edge\in\edges}\hat{\toll}_{\edge}\load_{\edge} + \increm\sum_{\edge\in\edges}\delta_{\edge}\load_{\edge}\\
&=
\sum_{\edge\in\edges}\hat{\toll}_{\edge}\load_{\edge} + \increm \sum_{\edge\in\edges}\delta_{\edge} \sum_{\route\ni\edge}\flow_{\route}\\
&=
\sum_{\edge\in\edges}\hat{\toll}_{\edge}\load_{\edge} + \increm 
\sum_{\pair\in\pairs}\sum_{\route\in\routes^{\pair}} \flow_{\route} \sum_{\edge\in\route} \delta_{\edge}\\
&=
\sum_{\edge\in\edges}\hat{\toll}_{\edge}\load_{\edge} + \increm 
\sum_{\pair\in\pairs}\sum_{\route\in\routes^{\pair}} \flow_{\route} \diffOD^{\pair}\\
&=
\sum_{\edge\in\edges}\hat{\toll}_{\edge}\load_{\edge} + \increm 
\sum_{\pair\in\pairs}\rate^{\pair} \diffOD^{\pair}.
\end{align*}
Since $\rate^{\pair}$ and $\diffOD^{\pair}$ are positive for all $\pair\in\pairs$, by choosing $\increm$ big enough, the quantity $\sum_{\edge\in\edges}\toll_{\edge}(\increm)\load_{\edge}$ can be made positive.
\end{proof}

\end{document}